\documentclass[letterpaper, 10 pt, conference]{ieeeconf}  % Comment this line out if you need a4paper

\IEEEoverridecommandlockouts                              % This command is only needed if 
\usepackage{graphics} % for pdf, bitmapped graphics files
\usepackage{epsfig} % for postscript graphics files
\usepackage[noadjust]{cite}
\usepackage{times} % assumes new font selection scheme installed
\usepackage{amsmath} % assumes amsmath package installed
\usepackage{amssymb}  % assumes amsmath package installed
\usepackage{amsfonts}
\usepackage{xcolor}
\usepackage{microtype}
\usepackage{cite}
\usepackage[hyphens]{url}
\usepackage{microtype}
\usepackage{mathtools}

\title{\LARGE \bf
Bridging Abstraction-Based Hierarchical Control and Moment Matching: A Conceptual Unification}

\author{Zirui Niu, Mohammad Fahim Shakib, and Giordano Scarciotti% <-this % stops a space
\thanks{Zirui Niu, Mohammad Fahim Shakib, and Giordano Scarciotti are with the Department of Electrical and Electronic Engineering, Imperial College London, SW7 2AZ, London, U.K. {\tt\footnotesize [zirui.niu20,m.shakib,g.scarciotti]@imperial.ac.uk}. }%
}

\newtheorem{thm}{\it Theorem}
\newtheorem{lem}{\it Lemma}
\newtheorem{assume}{\it Assumption}

\newtheorem{remark}{\it Remark}
\newtheorem{corol}{\it Corollary}

\newtheorem{define}{\it Definition}

\begin{document}

\maketitle
\thispagestyle{empty}
\pagestyle{empty}

%%%%%%%%%%%%%%%%%%%%%%%%%%%%%%%%%%%%%%%%%%%%%%%%%%%%%%%%%%%%%%%%%%%%%%%%%%%%%%%%
\begin{abstract}
In this paper, we establish a relation between approximate-simulation-based hierarchical control (ASHC) and moment matching techniques, and build a conceptual bridge between these two frameworks. To this end, we study the two key requirements of the ASHC technique, namely the bounded output discrepancy and the $M$-relation, through the lens of moment matching. We show that, in the linear time-invariant case, both requirements can be interpreted in the moment matching perspective through certain system interconnection structures. Building this conceptual bridge provides a foundation for cross-pollination of ideas between these two frameworks.
\end{abstract}

%%%%%%%%%%%%%%%%%%%%%%%%%%%%%%%%%%%%%%%%%%%%%%%%%%%%%%%%%%%%%%%%%%%%%%%%%%%%%%%%
\section{Introduction}\label{sec:intro}
In the field of computational science and engineering, real-world complex dynamical systems typically require large-scale mathematical models for accurate representations. However, the high dimensionality of those models can result in excessive computational costs for analysis and (control) design, leading to substantial processing time and energy consumption for, \textit{e.g.}, simulations and hardware implementations.
% These large-scale models arise in diverse applications, such as fluid dynamics, structural mechanics, circuit design, and control systems.
% While these models provide accurate representations of real-world phenomena, 
% However, the high dimensionality of those models can lead to excessive computational costs and time consumption. 
% The need for efficient computational methods becomes increasingly critical. 

Two prominent methodologies that address these challenges are model order reduction (MOR) and approximate-simulation-based hierarchical control (ASHC). MOR focuses on simplifying high-dimensional models while retaining essential system dynamics, see, \textit{e.g.},~\cite{ref:antoulas2005approximation,ref:schilders2008model,ref:scarciotti2024interconnection} and references therein for details of different MOR approaches. Instead, ASHC focuses on efficient control design for large-scale systems by decomposing control tasks into hierarchical levels. These levels comprise a simpler model, called abstract system, and an interface function that translates the input from the abstract system back to the original, large-scale system, as depicted in Fig.~\ref{fig:Hierarchical_Structure}. In such manner, the control synthesis of the large-scale system can be achieved with lower complexity~\cite{ref:girard2009hierarchical}. This ASHC approach has been applied, for instance, to control networked systems~\cite{ref:zamani2017compositional} and robot systems~\cite{ref:girard2009hierarchical}.

\begin{figure}[tbp]
\begin{centering}
    \includegraphics[height=4.5cm]{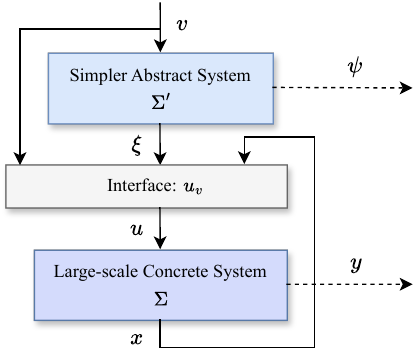}
    \caption{~~Hierarchical control system architecture. $\qquad\qquad\qquad\qquad\qquad$~~~~}
    \label{fig:Hierarchical_Structure}
\end{centering}
\end{figure}
\begin{figure}[tbp]
\begin{centering}
    \includegraphics[width=\linewidth]{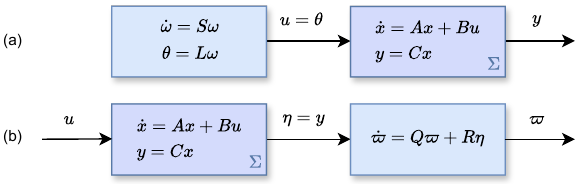}
    \caption{Illustration of the direct interconnection (a) and the swapped interconnection (b) in moment matching.}
    \label{fig:Interconnection_MM}
\end{centering}
\end{figure}

In this paper, we establish a relation between the ASHC technique introduced in~\cite{ref:girard2009hierarchical} and the moment-matching-based MOR technique introduced in~\cite{ref:astolfi2010model}. We observe that both methods hinge upon the use of specific interconnection structures. On the one hand, moment matching uses system interconnections, shown in Fig.~\ref{fig:Interconnection_MM}, to characterise some key properties, \textit{i.e.}, the so-called moments, to be preserved by the simpler, reduced-order model. This is achieved by studying the interconnection of the large-scale system with a signal generator (Fig.~\ref{fig:Interconnection_MM}(a)) or a filter (Fig.~\ref{fig:Interconnection_MM}(b)) that incorporate the information of the desired behaviours of interest~\cite{ref:astolfi2010model}. On the other hand, the hierarchical levels in the ASHC technique hinge on the interconnection structure in Fig.~\ref{fig:Hierarchical_Structure}, as already described. In studying these interconnections, both techniques rely on the solution of closely related Sylvester equations. This observation hints at a relation between these two methods.

To study this relation, we start with revisiting the fundamentals of moment matching and ASHC techniques. Then we reinterpret the ASHC technique through the lens of moment matching, showing that the two key design requirements proposed in the ASHC technique, namely the bounded output discrepancy and the so-called $M$-relation\footnote{This is usually called $\Pi$-relation in the literature~\cite{ref:pappas2000hierarchically}. We changed the letter for notational reasons.}, can be interpreted as moment matching requirements. Based on this result, we provide guidance for future works that aim, on the one hand, at extending various moment matching results, such as nonlinear and data-driven methods, to address ASHC problems, and, on the other hand, at establishing new techniques for control via reduced-order models. In summary, the main contribution of this paper lies in the establishment of a conceptual bridge between the two techniques, showing the possibility that the developed methods in each field can mutually benefit one another. Further works that explore these directions are already in preparation.

The remainder of this paper is organised as follows. Section~\ref{sec:Prelimiaries} provides a brief review of the two techniques. Section~\ref{sec:compare} gives the relation between ASHC and moment matching. An illustrative example is presented in Section~\ref{sec:example}, followed by some concluding remarks in Section~\ref{sec:concl}.

\textbf{Notation.} The symbols $\mathbb{C}_{0}$ and $\mathbb{C}_{<0}$ denote the set of complex numbers with zero real part and strictly negative real part, respectively, while $\mathbb{R}_{\geq 0}$ represents the set of non-negative real numbers. The symbol $I_{n}$ indicates an $n \times n$ identity matrix, $\sigma(A)$ represents the spectrum of the matrix $A \in \mathbb{R}^{n \times n}$, while for $B \in\mathbb{R}^{n\times m}$, $B^\top$ denotes its transpose. A square matrix $A$ is positive definite if it is symmetric and all its eigenvalues are positive. Given a bounded function $x: \mathbb{R} \rightarrow \mathbb{R}^{n}$, $\|x\|_{\infty}$ denotes its $L^{\infty}$ norm defined by $L^\infty \coloneqq \sup _{t \geq 0}\|x(t)\|$, where $\|\cdot\|$ indicates the Euclidean norm. A function $\gamma: \mathbb{R}_{\geq 0} \rightarrow \mathbb{R}_{\geq 0}$ belongs to class $\mathcal{K}$ if it is continuous, strictly increasing, and satisfies $\gamma(0)=0$.

\section{Review of Moment Matching and ASHC Techniques}\label{sec:Prelimiaries}
In this section, we recall the moment matching and ASHC techniques for a class of (large-scale) linear time-invariant (LTI) systems of the form
\begin{equation}\label{equ:ConcreteSystem}
\begin{aligned}
    \Sigma:\quad\begin{array}{l}
    \dot{x}=A x+B u, \qquad
    y=C x,
    \end{array}
\end{aligned}
\end{equation}
where $x(t) \in \mathbb{R}^{n}$, $u(t) \in \mathbb{R}^{m}$, and $y(t) \in \mathbb{R}^{p}$. The real-valued matrices $A$, $B$, and $C$ are of appropriate dimensions and, without loss of generality, assumed to satisfy $\operatorname{rank}(B)=m$ and $\operatorname{rank}(C)=p$. Then, both the moment matching and ASHC techniques aim to find a simpler LTI model in the form
\begin{equation}\label{equ:ReducedSystem}
\begin{aligned}
    \Sigma^{\prime}:\quad\begin{array}{l}
    \dot{\xi}=F \xi+G v, \qquad
    \!\psi=H \xi,
    \end{array}
\end{aligned}
\end{equation}
where $\xi(t) \in \mathbb{R}^{\hat{n}}$, $v(t) \in \mathbb{R}^{\hat{m}}$, and $\psi(t) \in \mathbb{R}^{p}$, and the real-valued matrices $F$, $G$, and $H$ are of appropriate dimensions, with order $\hat{n} \leq n$. Note that in the MOR technique, systems~(\ref{equ:ConcreteSystem}) and~(\ref{equ:ReducedSystem}) are called the \textit{full-} and the \textit{reduced-order model}, respectively, with the same number of inputs, \textit{i.e.}, $\hat{m} = m$. In the ASHC technique, instead, the simpler system~(\ref{equ:ReducedSystem}) is called the \textit{abstract system} or \textit{abstraction}, while the large-scale system~(\ref{equ:ConcreteSystem}) is called the \textit{concrete system}. The number of inputs of these two systems is not necessarily the same.
% To streamline our presentation, in the rest of the paper, we call system~(\ref{equ:ConcreteSystem}) the large-scale model and~(\ref{equ:ReducedSystem}) the simpler model.

\subsection{Moment Matching}\label{sec:PrelimiariesMM}
We first revisit the moment matching method, starting with the definition of moment of system~(\ref{equ:ConcreteSystem}).
\begin{define}[\cite{ref:scarciotti2024interconnection}]\label{def:moment}
    Consider system~(\ref{equ:ConcreteSystem}) and matrices $S \in \mathbb{R}^{\hat{n} \times \hat{n}}$ and $Q \in \mathbb{R}^{\hat{n} \times \hat{n}}$ with their spectra satisfying $\sigma(S) \cap \sigma(A) = \emptyset$ and $\sigma(Q) \cap \sigma(A) = \emptyset$. Let $L \in \mathbb{R}^{m \times \hat{n}}$ and $R \in \mathbb{R}^{\hat{n} \times p}$ be any two matrices such that the pair $(S, L)$ is observable and the pair $(Q, R)$ is reachable. Then
    the matrices $C\Pi$ and $\Upsilon B$ are called the \textit{moments} of system~(\ref{equ:ConcreteSystem}) at $(S,\, L)$ and $(Q,\, R)$ respectively, where $\Pi \in \mathbb{R}^{n \times \hat{n}}$ is the unique solution of the Sylvester equation
    \begin{equation}\label{equ:MomentDirectSylverter}
        \Pi S = A\Pi + BL,
    \end{equation}
    and $\Upsilon \in \mathbb{R}^{\hat{n} \times n}$ is the unique solution of the Sylvester equation
    \begin{equation}\label{equ:MomentSwappedSylverter}
        Q\Upsilon = \Upsilon A + RC.
        \vspace{-3mm}
    \end{equation}
    \hfill$\blacksquare$
\end{define}

\begin{remark}\label{rmk:interpolation}
    Originally moments were defined in the Laplace domain~\cite{ref:antoulas2005approximation}. More precisely, the $k$-th moment of a single-input single-output system~(\ref{equ:ConcreteSystem}) at an interpolation point $s \in \mathbb{C} \setminus \sigma(A)$ is defined as the $k$-th coefficient of the Laurent series expansion of the transfer function of system~(\ref{equ:ConcreteSystem}) at $s$. The moments computed at interpolation points equal to the eigenvalues of $S$ and $Q$ are nothing else that the elements of the matrices $C\Pi$ and $\Upsilon B$, see~\cite{ref:gallivan2006model,ref:astolfi2010model}, and~\cite{ref:shakib2023time} for the multiple-input, multiple-output case.
    % \red{In~\cite{ref:gallivan2006model}, the connection between the notion of moment and the Sylvester equation was established, whereas the relation with steady-state was established in~\cite{ref:astolfi2010model}.}

    % Both the definition and the matching of moments for multi-input, multi-output systems are more complicated, as a series of vectors should be defined to guarantee the right tangential direction interpolations of the transfer-function matrix, see~\cite{ref:ionescu2013families,ref:ionescu2014families,ref:shakib2023time} for more details.
\end{remark}

It is possible to show~\cite{ref:astolfi2010model} that the model
\begin{equation}\label{equ:ReducedMMSL}
    \dot{\xi} = (S - GL) \xi + G u, \qquad
    \psi = C\Pi \xi
\end{equation}
matches the moments of~(\ref{equ:ConcreteSystem}) (meaning that the moments of the two models are identical) at $(S,\,L)$ for any $G$ such that $\sigma(S) \cap \sigma(S- G L)=\emptyset$. Likewise, the model
\begin{equation}\label{equ:ReducedMMQR}
    \dot{\xi} = (Q - RH) \xi + \Upsilon B u, \qquad
    \psi = H \xi
\end{equation}
matches the moments of~(\ref{equ:ConcreteSystem}) at $(Q,\, R)$ for any $H$ such that $\sigma(Q-R H) \cap \sigma(Q) = \emptyset$. Moreover, a model that matches the moments of~(\ref{equ:ConcreteSystem}) at $(S,\,L)$ and $(Q,\, R)$ simultaneously is called two-sided moment matching model. This two-sided matching can be achieved, for example, by model~(\ref{equ:ReducedMMSL}) with $G = (\Upsilon \Pi)^{-1} \Upsilon B$% \begin{equation}\label{equ:twoSidedMMG}
%     G=(\Upsilon \Pi)^{-1} \Upsilon B.
% \end{equation}
, or by model~(\ref{equ:ReducedMMQR}) with $H = C \Pi(\Upsilon \Pi)^{-1}$, provided that $\Upsilon \Pi$ is invertible, see~\cite{ref:scarciotti2024interconnection} for more details.

Inspired by the regulator theory~\cite{ref:isidori1985nonlinear}, the Sylvester equations~(\ref{equ:MomentDirectSylverter}) and~(\ref{equ:MomentSwappedSylverter}) also induce a connection between moments and the steady-state output responses of two interconnection structures, namely the direct and swapped interconnections, as depicted in Fig.~\ref{fig:Interconnection_MM} and %. The relation between moments and steady-state output responses of the interconnection structures is 
summarized as follows.

% \footnote{This means the system is driven by the output of a signal generator.}
% \footnote{This means the output of the system drives a filter.}

\begin{thm}[\cite{ref:scarciotti2024interconnection}] \label{thm:SteadyState}
Consider system~(\ref{equ:ConcreteSystem}) and suppose $\sigma(A) \subset \mathbb{C}_{<0}$. 
% Given a set of interpolation points $\mathcal{P} = \{s_1, s_2, \ldots, s_\nu\} \subset \mathbb{C}_0$ with elements $\{s_i\}_{i=1}^{\nu}$ distinct, 
Let $S \in \mathbb{R}^{\hat{n} \times \hat{n}}$ and $Q \in \mathbb{R}^{\hat{n} \times \hat{n}}$ be any two 
% non-derogatory\footnote{A matrix is non-derogatory if its characteristic and minimal polynomial coincide}
matrices with simple eigenvalues satisfying $\sigma(S) \subset \mathbb{C}_{0}$ and $\sigma(Q) \subset \mathbb{C}_{0}$ with $\sigma(S) \cap \sigma(Q) = \emptyset$. Let $L \in \mathbb{R}^{m \times \hat{n}}$ and $R \in \mathbb{R}^{\hat{n} \times p}$ be any two matrices such that the pair $(S, L)$ is observable and the pair $(Q, R)$ is reachable. Then the following statements hold.
\begin{itemize}
    \item The moment at $(S, L)$ admits a one-to-one relation with the steady-state of output $y$ of the (direct) interconnection  between system~(\ref{equ:ConcreteSystem}) and the signal generator
    \begin{equation}
    \label{eq-genS}
        \dot{\omega} =S \omega, \quad \theta=L \omega,
    \end{equation}
    via $u = \theta$ (Fig.~\ref{fig:Interconnection_MM}(a)), provided $(S,\, \omega(0))$ is excitable~\cite{ref:padoan2017geometric}.
    \item The moment at $(Q, R)$ admits a one-to-one relation with the steady-state of output $\varpi$ of the (swapped) interconnection between system~(\ref{equ:ConcreteSystem}) and the filter
    \begin{equation}
    \label{eq-filQ}
        \dot{\varpi} = Q \varpi+R \eta
    \end{equation}
   via $\eta = y$ (Fig.~\ref{fig:Interconnection_MM}(b)), for $x(0) = \varpi(0) = 0$ and any non-zero signal $u$ that exponentially decays to zero.\hfill$\blacksquare$
\end{itemize}
\end{thm}

\begin{figure}[tbp]
\begin{centering}
    \includegraphics[width=0.82\linewidth]{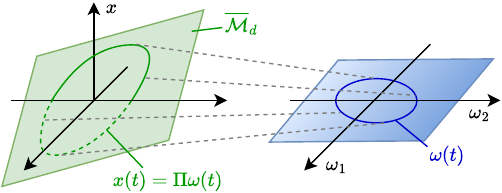}
    \caption{Visualisation of the Subspace $\overline{\mathcal{M}}_d$.}
    \label{fig:Manifold}
\end{centering}
\end{figure}

Taking a step back, the mere satisfaction of~(\ref{equ:MomentDirectSylverter}) and~(\ref{equ:MomentSwappedSylverter}) (without the additional assumptions in Theorem~\ref{thm:SteadyState}) implies the existence of two invariant subspaces, namely $\overline{\mathcal{M}}_d = \{(x, \omega) \in \mathbb{R}^{n+\hat{n}} \,|\, x = \Pi\omega \}$ in Fig.~\ref{fig:Interconnection_MM}(a) when $u = \theta$; and $\overline{\mathcal{M}}_s = \{(x, \varpi) \in \mathbb{R}^{n+\hat{n}} \,|\, \varpi = -\Upsilon x \}$ in Fig.~\ref{fig:Interconnection_MM}(b) when $u \equiv 0$. Moreover, under the assumptions of Theorem~\ref{thm:SteadyState}, it is possible to show~\cite{ref:simpson2024steady} that the dynamics on the globally exponentially attractive subspace $\overline{\mathcal{M}}_d$, visualised in Fig.~\ref{fig:Manifold}, are described by
\begin{equation}\label{eq-inv1}
\dot \omega = S \omega, \quad\text{ and }\quad    y = C\Pi \omega.
\end{equation}
Similarly, the dynamics of the error variable $\zeta = \varpi + \Upsilon x$ are described by
\begin{equation}\label{eq-inv2}
\dot \zeta = Q \zeta + \Upsilon B u.
\end{equation}
Thus, $C\Pi$ is the steady-state observation matrix relating $y$
to the state $\omega$ of the signal generator, while $\Upsilon B$ influences how the control $u$
impacts the deviation from the invariant subspace $\overline{\mathcal{M}}_s$~\cite{ref:simpson2024steady}.
%\red{When $u \nequiv 0$ in} the interconnection in Fig.~\ref{fig:Interconnection_MM}(b), \blue{the moment $\Upsilon B$ of system~(\ref{equ:ConcreteSystem}) influences how the control $u$ impacts the deviation from the invariant set $\bar{\mathcal{M}}_s$~\cite{ref:simpson2024steady}}. 
Then one can see that the intuition behind the reduced-order model (\ref{equ:ReducedMMSL}) is that when $u$ is generated by (\ref{eq-genS}) and $A$ and $S-GL$ are Hurwitz, then  \eqref{equ:ConcreteSystem} and (\ref{equ:ReducedMMSL}) have the same asymptotic behaviour, which is described by (\ref{eq-inv1}). A similar intuition holds for the swapped interconnection. 
%Theorem~\ref{thm:SteadyState} implies that the \red{reduced-order} model must produce the same steady-state output response when being injected into the same interconnection framework shown in Fig.~\ref{fig:Interconnection_MM}. 
These observations are important for reinterpreting the ASHC technique in the next section.

In the context of model reduction, the benefit brought by Theorem~\ref{thm:SteadyState} lies in extending the notion of moment to general classes of systems, such as nonlinear and time-delay systems, see  \textit{e.g.}~\cite{ref:scarciotti2015model}, and providing data-driven approaches through the use of the invariant manifolds $\overline{\mathcal{M}}_d$ and $\overline{\mathcal{M}}_s$~\cite{ref:scarciotti2017data,ref:mao2024data}.

\subsection{Approximate-Simulation-based Hierarchical Control}\label{sec:PrelimiariesASHC}
Now we turn our focus to the ASHC technique, recalling its two key requirements. The ASHC technique aims to design a simpler abstract system and an interface function $u_{v}: \mathbb{R}^{\hat{m}}\times \mathbb{R}^{\hat{n}} \times \mathbb{R}^{n} \rightarrow \mathbb{R}^{m}$ to achieve the control synthesis of the large-scale system with lower computational costs, see Fig.~\ref{fig:Hierarchical_Structure}. To provide guarantees on the trajectories of the large-scale system, the first key requirement is that the error between the output trajectories of the large-scale and simpler systems is bounded through a so-called ``simulation function''~\cite{ref:girard2007approximation,ref:girard2009hierarchical}, defined as follows.

\begin{define}[\cite{ref:girard2009hierarchical}]\label{def:simuFunction}
Let $V\!:\! \mathbb{R}^{\hat{n}} \!\times\! \mathbb{R}^n \!\rightarrow\! \mathbb{R}_{\geq 0}$ be a smooth function and $u_{v}$ \!:\! $\mathbb{R}^{\hat{m}} \!\times\! \mathbb{R}^{\hat{n}} \!\times\! \mathbb{R}^n \!\rightarrow\! \mathbb{R}^m$ be a continuous function. $V$ is a \textit{simulation function} of system~(\ref{equ:ReducedSystem}) by system~(\ref{equ:ConcreteSystem}), and $u_v$ is an associated \textit{interface}, if there exists a class-$\mathcal{K}$ function $\gamma$ such that for all $(\xi, x) \in \mathbb{R}^{\hat{n}} \times \mathbb{R}^n$, $V(\xi, x) \geq\|H\xi - Cx\|$,
and for all $v \in \mathbb{R}^{\hat{m}}$ satisfying $\gamma(\|v\|)<V(\xi, x)$, then
$
\frac{\partial V(\xi, x)}{\partial \xi} (F\xi \!+\! G v) \!+\!\frac{\partial V(\xi, x)}{\partial x}\!\left(A x \!+\! B u_v(v, \xi, x)\right)<0.
$
\hfill$\blacksquare$
\end{define}

In fact, the above simulation and interface functions guarantee that the error between the output trajectories $\psi$ and $y$ of two systems satisfies 
$
    \|\psi(t)-y(t)\| \leq \max \left\{V(\xi(0), x(0)), \gamma\left(\|v\|_{\infty}\right)\right\}
$
for all $t \geq 0$, see~\cite[Theorem 1]{ref:girard2009hierarchical} for more details.

Before looking at how to design the simulation function and the interface, the following standard assumption is introduced.
\begin{assume}\label{asmp:stabilisability}
    System~(\ref{equ:ConcreteSystem}) is stabilisable.
\end{assume}

Then we have the following results.
\begin{lem}[\cite{ref:girard2009hierarchical}]
Suppose Assumption~\ref{asmp:stabilisability} holds. Then there exists a positive definite symmetric matrix $W \in \mathbb{R}^{n \times n}$ and a strictly positive scalar $\lambda$ such that 
% $W \succeq C^{\top} C$ and 
% $ (A+B K)^{\top} W+W(A+B K) \preceq -2 \lambda W $
\begin{equation*}
    W \succeq C^{\top} C, \quad
    (A+B K)^{\top} W+W(A+B K) \preceq -2 \lambda W ,
\end{equation*}
with $K \!\in\! \mathbb{R}^{m \times n}$ any matrix such that $\sigma(A \!+\! BK) \!\subset\! \mathbb{C}_{<0}$. \hfill$\blacksquare$
\end{lem}

\begin{thm}[\cite{ref:girard2009hierarchical}]\label{thm:simuFunction}
Suppose Assumption~\ref{asmp:stabilisability} holds and assume there exist matrices $P \in \mathbb{R}^{n \times \hat{n}}$ and $\hat{L} \in \mathbb{R}^{m \times \hat{n}}$ solving the linear matrix equations
\begin{subequations}\label{equ:SimuCondition}
    \begin{align}
     P F &= A P + B \hat{L}, \label{equ:SimuSylvester} \\
     H &= C P . \label{equ:SimuOutput}
    \end{align}
\end{subequations}
Then, the function $ V(\xi, x) \coloneqq \sqrt{(P \xi-x)^{\top} W(P \xi-x)}$
% \begin{equation}\label{equ:SimuFunctionLTI}
%     V(\xi, x) \coloneqq \sqrt{(P \xi-x)^{\top} W(P \xi-x)}
% \end{equation}
is a simulation function of system~(\ref{equ:ReducedSystem}) by~(\ref{equ:ConcreteSystem}) and the associated interface takes the form
\begin{equation}\label{equ:interface}
    u_v(v, \xi, x) \coloneqq \hat{R} v + \hat{L}\xi + K(x-P\xi),
\end{equation}
for any matrix $\hat{R} \in \mathbb{R}^{m \times \hat{m}}$.\hfill$\blacksquare$
\end{thm}

The second key requirement of the ASHC technique focuses on recovering any output trajectory generated by the large-scale system. This requirement is formulated as ``for any output of the concrete system, we can construct an input of the abstraction such that this produces the same output''. This is reminiscent, but different, of moment matching, which instead requires that
%Unlike MOR techniques, which require 
the simpler model possess the same steady state when subject to the same input generated by \eqref{eq-genS}. %, the second key requirement of the ASHC technique focuses on recovering any output trajectory generated by the large-scale system through the simpler system. 
The ASHC requirement is specified as follows.
\begin{define}[\cite{ref:girard2009hierarchical}]\label{def:PiRelation}
System~(\ref{equ:ReducedSystem}) is \textit{$M$-related} to system~(\ref{equ:ConcreteSystem}) if there exists a matrix $M \in \mathbb{R}^{\hat{n} \times n}$ such that for all $x \in \mathbb{R}^n$ and $u \in \mathbb{R}^m$, there exists $v \in \mathbb{R}^{\hat{m}}$ satisfying
\begin{subequations}\label{equ:PiRelation}
\begin{align}
    M(A x+B u) &= F M x+G v, \label{equ:PiRelationState} \\
    C &= H M. \label{equ:PiRelationOutput}
\end{align}
\end{subequations}
\end{define}
\hfill$\blacksquare$

Note that the $M$-relation yields that for any state trajectories of system~(\ref{equ:ConcreteSystem}) with any input function $u$, there exists an input function $v$ to the simpler system~(\ref{equ:ReducedSystem}) such that the subspace $\mathcal{M}_s = \{(x, \xi) \,|\, \xi = M x \}$ is invariant, see~\cite{ref:pappas2000hierarchically} for more details. Then the following result holds on the output level.
\begin{thm}[\cite{ref:girard2009hierarchical}]\label{thm:PiRelationOutputMatch}
    If system~(\ref{equ:ReducedSystem}) is $M$-related to~(\ref{equ:ConcreteSystem}), then for any output trajectory $y$ of system~(\ref{equ:ConcreteSystem}), $\psi=y$ is an output trajectory of system~(\ref{equ:ReducedSystem}) for some input function $v$. \hfill$\blacksquare$
\end{thm}

To summarize, the ASHC technique for the large-scale system~(\ref{equ:ConcreteSystem}) is solved if one can design a simpler system~(\ref{equ:ReducedSystem}) and an interface~(\ref{equ:interface}) that satisfies both~(\ref{equ:SimuCondition}), which guarantees the boundedness of the error between output trajectories of the two systems, and~(\ref{equ:PiRelation}), which avoids the loss of controllable behaviours of system~(\ref{equ:ConcreteSystem}) when constructing the abstract system~(\ref{equ:ReducedSystem}). These two requirements are viewed next through the lens of moment matching. 

\section{ASHC in the Moment Matching Framework}\label{sec:compare}
In this section, we establish the connection between moment matching and ASHC techniques by reinterpreting the two key requirements posed by the ASHC technique explained in Section~\ref{sec:Prelimiaries} through the lens of the direct and swapped interconnections in Fig.~\ref{fig:Interconnection_MM}. To this end, we first look at the bounded output discrepancy requirement and study the $M$-relation requirement after that.

\subsection{Bounded Output Discrepancy}\label{sec:compareOutputBound}
The first design requirement of the ASHC technique is that the discrepancy between the output trajectories of the large-scale system~(\ref{equ:ConcreteSystem}) controlled by an interface~(\ref{equ:interface}), and the simpler system~(\ref{equ:ReducedSystem}) is bounded by a simulation function. Recall that Theorem~\ref{thm:simuFunction} yields that, under Assumption~\ref{asmp:stabilisability}, the output discrepancy remains bounded for any input $v \in L^\infty$ if there exist $P \in \mathbb{R}^{n \times \hat{n}}$ and $\hat{L} \in \mathbb{R}^{m \times \hat{n}}$ solving~(\ref{equ:SimuCondition}). 

Since Theorem~\ref{thm:SteadyState} points out the one-to-one relations between moments and steady-state responses, we reinterpret equations~(\ref{equ:SimuCondition}), together with the interface function~(\ref{equ:interface}), in the moment matching framework. To this end, for specific configurations, the hierarchical structure shown in Fig.~\ref{fig:Hierarchical_Structure} can be regarded as a cascade interconnection similar to the ones in Fig.~\ref{fig:Interconnection_MM}. Hereto, denote the interface $u_v$ in~(\ref{equ:interface}) as $u_v = \hat{R}v + \hat{L}\xi + K(x-P\xi) = Kx + u_v^{*}$ with $u_v^{*} = \hat{R}v + (\hat{L}- KP)\xi$. Then the hierarchical structure in Fig.~\ref{fig:Hierarchical_Structure} is equivalent to the interconnection shown in Fig.~\ref{fig:Interconnect_Output}, where system~(\ref{equ:ReducedSystem}) is cascaded with the system
\begin{equation}\label{equ:ConcreteSysStable}
    \dot{x} = (A + BK) x + B u, \quad y = C x,
\end{equation}
via the link $u = u_v^{*} = \hat{R}v + (\hat{L} - KP)\xi$, with $K$ such that $\sigma(A +BK) \subset \mathbb{C}_{<0}$. We first study the case for $v \equiv 0$, for which we have the following result.
% Given a set of interpolation points $\mathcal{P}_F = \{s_1, s_2, \ldots, s_m\} \subset \mathbb{C}_{0}$ with elements $\{s_i\}_{i=1}^{\hat{n}}$ distinct, denote $F \in \mathbb{R}^{n \times n}$ such that $\sigma(F) = \mathcal{P}_F$. 

\begin{figure}[tbp]
\begin{centering}
    \includegraphics[width=0.85\linewidth]{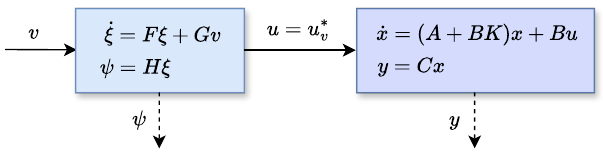}
    \caption{Interconnection for the interpretation of the linear matrix equations~(\ref{equ:SimuCondition}) with a stabilising link $u^{*}_v = \hat{R}v + (\hat{L}- KP)\xi$.}
    \label{fig:Interconnect_Output}
\end{centering}
\end{figure}

\begin{thm}\label{thm:hierarchicalMatching}
    Consider the hierarchical structure shown in Fig.~\ref{fig:Hierarchical_Structure} and, equivalently, Fig.~\ref{fig:Interconnect_Output}. Suppose $v \equiv 0$ and Assumption~\ref{asmp:stabilisability} holds. Let $F \in \mathbb{R}^{\hat{n} \times \hat{n}}$ be any matrix with simple eigenvalues satisfying $\sigma(F) \subset \mathbb{C}_0$ and $\sigma(F) \cap \sigma(A) = \emptyset$. Let $\hat{L}$ and $K$ in~(\ref{equ:interface}) be such that $(F,\, \hat{L})$ is observable and $\sigma(A +BK) \subset \mathbb{C}_{<0}$. Let $P \in \mathbb{R}^{n \times \hat{n}}$ be the unique solution of~(\ref{equ:SimuSylvester}) and let the pair $(F,\, \xi(0))$ be excitable. Then, there exists a one-to-one relation between the moment $CP$ of system~(\ref{equ:ConcreteSystem}) at $(F,\, \hat{L})$ and the steady-state output $y$ of system~(\ref{equ:ConcreteSystem}) in Fig.~\ref{fig:Interconnect_Output}. 
\end{thm}
\begin{proof}
    Following Definition~\ref{def:moment}, as $\sigma(F) \cap \sigma(A) = \emptyset$ and $(F,\, \hat{L})$ is observable, $CP$ is the moment of system~(\ref{equ:ConcreteSystem}) at $(F, \hat{L})$.
    When $v \equiv 0$, the hierarchical structure in Fig.~\ref{fig:Interconnect_Output}, can be regarded as a direct interconnection formed by the generator~(\ref{equ:ReducedSystem}) and system~(\ref{equ:ConcreteSysStable}). Since $\sigma(F) \cap \sigma(A + BK) = \emptyset$ and the pair $(F, \hat{L} - KP)$ is observable (which we prove later), $C \overline{P}$ denotes the moment of~(\ref{equ:ConcreteSysStable}) at $(F, \hat{L} - KP)$, where $\overline{P} \in \mathbb{R}^{n \times \hat{n}}$ is the unique solution to
    \begin{equation}\label{equ:SimuSylversterClosed}
        \overline{P} F = (A + B K) \overline{P} + B (\hat{L} - KP),
    \end{equation}
    with $P$ solving~(\ref{equ:SimuSylvester}), see Definition~\ref{def:moment}. By subtracting~(\ref{equ:SimuSylvester}) from~(\ref{equ:SimuSylversterClosed}), we obtain
    $(\overline{P} - P) F = A (\overline{P} - P) + B\hat{L} + BK(\overline{P} - P) - B\hat{L} = (A + BK)(\overline{P} - P).$
    % \begin{equation*}
    % \begin{aligned}
    %     (\overline{P} - P) F &= A (\overline{P} - P) + B\hat{L} + BK(\overline{P} - P) - B\hat{L} \\
    %     &= (A + BK)(\overline{P} - P).
    % \end{aligned}
    % \end{equation*}
    Since $\sigma(F) \cap \sigma(A + BK) = \emptyset$, it follows that $\overline{P} = P$. Hence, the moment $C\overline{P}$ of system~(\ref{equ:ConcreteSysStable}) at $(F, \hat{L} - KP)$ coincides with the moment $CP$ of system~(\ref{equ:ConcreteSystem}) at $(F, \hat{L})$, and, by Theorem~\ref{thm:SteadyState}, has a one-to-one relation with the steady-state output $y$ with $v\equiv 0$ as the pair $(F, \hat{L} - KP)$ is observable. To show this observability property, we assume, by contradiction, that the pair $(F, \hat{L} - KP)$ is not observable, \textit{i.e.}, by the PBH test, there exists a non-zero vector $\varepsilon \in \mathbb{R}^{\hat{n}}$ such that $F\varepsilon = \lambda_F \varepsilon$ and $(\hat{L} - KP)\varepsilon = 0$ with $\lambda_F \in \sigma(F)$. Then right-multiplying $\varepsilon$ to both sides of~(\ref{equ:SimuSylversterClosed}) implies $(\lambda_F I_n - A - BK)\overline{P}\varepsilon = B (\hat{L} - KP)\varepsilon = 0$. Since $\sigma(F) \cap \sigma(A + BK) = \emptyset$, $\operatorname{rank} (\lambda_F I_n - A - BK) = n$ and therefore $\overline{P}\varepsilon = P\varepsilon = 0$. Consequently, $(\hat{L} - KP)\varepsilon = \hat{L}\varepsilon = 0$, which, by PBH test, contradicts the observability of $(F, \hat{L})$. Therefore, the pair $(F, \hat{L} - KP)$ is observable, and this concludes the proof.
    % Meanwhile, since the pair $(F, \hat{L})$ is observable, which can be proved by the PBH test under the conditions $\sigma(F) \cap \sigma(A) = \emptyset$ and 
    % Then Theorem~\ref{thm:SteadyState} implies that the moment $C \overline{P}$ has a one-to-one relation with the steady-state output $y$ of the direct interconnection shown in Fig.~\ref{fig:Interconnect_Output}.
    % We have thus proven that the moment $CP$ of system~(\ref{equ:ConcreteSystem}) at $(F, \hat{L})$ has a one-to-one relation with the steady-state output $y$ of system~(\ref{equ:ConcreteSystem}) in the hierarchical structure under $v\equiv 0$. 
\end{proof}

% To see this, assume $A$ in~(\ref{equ:ConcreteSystem}) is Hurwitz, and we have a set of interpolation points $\mathcal{P}_F = \{s_1, s_2, \ldots, s_m\} \subset \mathbb{C} \setminus \sigma(A)$. If $F$ is selected such that $\sigma(F) = \mathcal{P}_F$, then a valid inference~(\ref{equ:interface}) can be designed with $K = 0$, $R = 0$, and $\hat{L}$ such that $(F,\, \hat{L})$. Consequently, the satisfaction of~(\ref{equ:SimuCondition}) suggests the one-to-one relation between the moments at $\mathcal{P}_F$ and the steady-state response of the direct interconnection formed by the signal generator of the form~(\ref{equ:ReducedSystem}) with $G = 0$ and the concrete system~(\ref{equ:ConcreteSystem}) with $u = u_v = \hat{L}\xi$. The structure is depicted in Fig.~\ref{fig:Interconnect_Output}. 

The result of Theorem~\ref{thm:hierarchicalMatching} provides a moment matching perspective of the bounded output discrepancy requirement. Following the proof of Theorem~\ref{thm:hierarchicalMatching}, system~(\ref{equ:ReducedSystem})  with $G = 0$ is the same as system~(\ref{eq-inv1}), for $S=F$ and $L=\hat{L}$ (and so $\Pi = P$). Thus, the abstract system with $H=CP$ behaves as a %both a signal generator in the direct interconnection in Fig.~\ref{fig:Interconnect_Output} and a ``pseudo
``\textit{limiting}'' reduced-order model. %as \red{in~(\ref{equ:ReducedMMSL}) with $G = 0$ and $S = F$}. Then 
It should be stressed that the abstract system is not of the form (\ref{equ:ReducedMMSL}), because $G=0$, and consequently it does not match the moments $CP$ at $(F,\hat{L})$. But, as explained in Section~\ref{sec:PrelimiariesMM}, the limiting system~(\ref{eq-inv1}) is the system towards which the reduced-order model~(\ref{equ:ReducedMMSL}) tends to when $u$ is driven by~(\ref{eq-genS}) via $u = \theta$ and $S-GL$ is Hurwitz. Thus, even though the abstract system is not a reduced-order model by moment matching, it still has the property of matching, for $v\equiv 0$,
%As explained in Section~\ref{sec:PrelimiariesMM}, 
%the steady-state output \red{$\psi = CP \xi$} of this ``reduced-order model'' matches 
the steady-state output $y$ of the large-scale system~(\ref{equ:ConcreteSystem}).
%in the direct interconnection. 

% given a set of interpolation points $\mathcal{P}_F = \{s_1, s_2, \ldots, s_m\} \subset \mathbb{C} \setminus \sigma(A)$, if $F$ is selected such that $\sigma(F) = \mathcal{P}_F$, then~(\ref{equ:SimuCondition}) and~(\ref{equ:interface}) indicates that the abstract system~(\ref{equ:ReducedSystem}) with $G = 0$ is a ``reduced-order model'' that match the moments of the concrete system~(\ref{equ:ConcreteSystem}) at $(F,\, \hat{L})$ for some $\hat{L}$ such that the pair $(F,\, \hat{L})$ is observable. 

% Motivated by Theorem~\ref{thm:SteadyState} that points out the one-to-one match between moments and steady-state response, equations~(\ref{equ:SimuCondition}) and~(\ref{equ:interface}) can also be interpreted from the steady-state perspective based on the interconnection framework. To see this, 

% \begin{equation*}
%     \dot{\omega_f} = F \omega_f, \quad \theta_f = \hat{L} \omega_f,
% \end{equation*}
% with $\omega_f(t) \in \mathbb{R}^{\hat{n}}$ and $\theta_f(t) \in \mathbb{R}^{m}$, 

% In other words, this steady-state analysis yields that if $\sigma(A) \subset \mathbb{C}_{<0}$ and the input to abstract system $v \equiv 0$, the distance between outputs of systems~(\ref{equ:ConcreteSystem}) and~(\ref{equ:ReducedSystem}), connected in the hierarchical structure shown in Fig.~\ref{fig:Hierarchical_Structure}, converges to zero even when matrices $G$ in~(\ref{equ:ConcreteSystem}) and $R$ in~(\ref{equ:interface}) are non-zero.

Continuing with this steady-state matching analysis, we now study the case %finally study the hierarchical structure in Fig.~\ref{fig:Hierarchical_Structure} (or equivalently, the interconnection in Fig.~\ref{fig:Interconnect_Output}) 
when the input $v \not\equiv 0$. By denoting $e_s = x - P \xi$, the satisfaction of~(\ref{equ:SimuCondition}) and~(\ref{equ:interface}) implies
\begin{equation}\label{equ:directStateError}
    \!\!\!\!\!\!\begin{array}{rl}
    \dot{e}_s\!\!\!\!\! &= Ax + Bu_v - PF\xi - PGv \\
        &= Ae_s \!+\! AP\xi \!+\! B\hat{R}v \!+\! B\hat{L}\xi \!+\! BKe_s \!-\! PF\xi \!-\! PGv \\
        &= (A + BK)e_s + (B\hat{R} - PG)v,
    \end{array}
\end{equation}
where $K$ is such that $\sigma(A + BK) \subset \mathbb{C}_{<0}$. Now by~(\ref{equ:SimuOutput}), the discrepancy in the output trajectories of two systems satisfies 
$
e_y = y - \psi = Cx - H\xi = Ce_s.
$
On the one hand, this relation coincides with the previous statement that when $v \equiv 0$, the set $\mathcal{M}_d = \{(x, \xi)\, |\, x = P\xi \}$ is exponentially attractive and invariant for all $(x(0), \xi(0)) \in \mathbb{R}^{n+\hat{n}}$. This results in the exponential decay of $e_y$ to zero and the matching of the steady states of the two output trajectories $\psi$ and $y$. On the other hand, when $v \not\equiv 0$, exponential stability of~(\ref{equ:directStateError}) implies that $e_s$ (and $e_y$) are bounded for any input $v \in L^\infty$. This result is consistent with the target of the ASHC technique. In summary, when $v \equiv 0$, the steady-state outputs are matched, and when $v \not\equiv 0$, the two outputs have a bounded error, which can be minimized by an opportune selection of $\hat{R}$ in~(\ref{equ:interface}).

\begin{remark}
    The matching of the steady-state outputs of the two systems~(\ref{equ:ConcreteSystem}) and~(\ref{equ:ReducedSystem}) when $v \equiv 0$ is consistent with Definition~\ref{def:simuFunction}, because the simulation function guarantees that $||\psi(t) \!-\! y(t)||$ decays to zero as $\gamma(\|v(t)\|) \!=\! 0$ for all times.
\end{remark}
% \begin{remark}\label{rmk:simuMinBound}
%     When $\operatorname{rank}(B)=m$, the magnitude of the bound to the output discrepancy $e_y$ can be minimised by \red{choosing $\hat{R}$ in~(\ref{equ:interface}) as} $\hat{R}=(B^{\GS{\top}} W B)^{-1} B^{\GS{\top}} W P G$ with $W$ solving~(\ref{equ:simuLMI}), see~\cite{ref:girard2009hierarchical} for more details.
% \end{remark}

\subsection{$M$-relation}\label{sec:comparePiRelation}
Now we discuss the second key design requirement in the ASHC technique: the $M$-relation. To this end, we first show that the $M$-relation requirement~(\ref{equ:PiRelation}) can also be characterised by Sylvester equations. 

\begin{thm}\label{thm:PiRelationSyl}
System~(\ref{equ:ReducedSystem}) is $M$-related to system~(\ref{equ:ConcreteSystem}) if there exist matrices $N \in \mathbb{R}^{\hat{m} \times n}$ and $\Gamma \in \mathbb{R}^{\hat{m} \times m}$ solving
\begin{subequations}\label{equ:PiRelationV2}
\begin{align}
    M A &= F M + G N, \label{equ:PiRelationSyl} \\
    G \Gamma &= M B, \label{equ:PiRelationInput} \\
    C &= H M, \label{equ:PiRelationOutputV2}
\end{align}
\end{subequations}
for some matrix $M \in \mathbb{R}^{\hat{n} \times n}$.
\end{thm}
\begin{proof}
    Consider the right side of (\ref{equ:PiRelationState}). By substituting $v = Nx + \Gamma u$, we have
   $ F M x + G v =  F M x + GNx + G\Gamma u = M A x + M B u  = M(Ax + Bu)$, where we have used \eqref{equ:PiRelationSyl} and \eqref{equ:PiRelationInput}.
    % \begin{equation*}
    %     \begin{aligned}
    %          F M x + G v &=  F M x + GNx + G\Gamma u = M A x + M B u \\
    %          & = M(Ax + Bu).
    %     \end{aligned}
    % \end{equation*}
    From here, we conclude that~(\ref{equ:PiRelationV2}) implies~(\ref{equ:PiRelation}), proving the claim.
\end{proof}

Theorem~\ref{thm:PiRelationSyl} proposes a sufficient condition for the $M$-relation based on a Sylvester equation and the restriction that the input $v$ that drives system~(\ref{equ:ReducedSystem}) takes the form 
\begin{equation}\label{equ:PiRelationLink}
   v = Nx + \Gamma u.
\end{equation}
This restriction, in fact, is naturally satisfied by the existing geometric method proposed by~\cite{ref:girard2009hierarchical} for finding a valid simpler system that satisfies both~(\ref{equ:SimuCondition}) and~(\ref{equ:PiRelation}). The method is summarised as follows.

\begin{figure}[tbp]
\begin{centering}
    \includegraphics[width=0.85\linewidth]{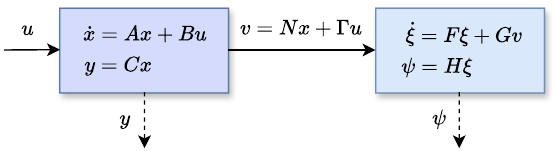}
    \caption{Interconnection for the interpretation of the linear matrix equations~(\ref{equ:PiRelationV2}).}
    \label{fig:Interconnect_PiRelation}
\end{centering}
\end{figure}

\begin{thm}[\cite{ref:girard2009hierarchical}]\label{thm:abstractDesign}
Consider the systems~(\ref{equ:ConcreteSystem}) and~(\ref{equ:ReducedSystem}). Let
\begin{itemize}
    \item $P \in \mathbb{R}^{n \times \hat{n}}$ be an injective map such that $\operatorname{im}(A P) \subseteq \operatorname{im}(P)+\operatorname{im}(B)$ and $\operatorname{im}(P)+\operatorname{ker}(C)=\mathbb{R}^n$;
    \item $D \in \mathbb{R}^{n \times (\hat{m} -m)}$, $E \in \mathbb{R}^{(\hat{m} - m) \times n}$, and $M \in \mathbb{R}^{\hat{n} \times n}$ satisfy $\operatorname{im}(D) \subseteq \operatorname{ker}(C)$, $\operatorname{im}(P) \oplus \operatorname{im}(D)=\mathbb{R}^n$, $M P = I_{\hat{n}}$, and $P M+D E = I_n$;
    \item $F \in \mathbb{R}^{\hat{n} \times \hat{n}}$ and $\hat{L} \in \mathbb{R}^{m \times \hat{n}}$ be such that $AP \!=\! PF \!-\! B\hat{L}$;
    \item $H=C P$ and $G=\left[MB \ \, MAD \right]$.
\end{itemize}
Then,~(\ref{equ:SimuCondition}) and (\ref{equ:PiRelation}) hold and therefore system~(\ref{equ:ReducedSystem}) is $M$-related to~(\ref{equ:ConcreteSystem}). \hfill$\blacksquare$
\end{thm}

% \begin{equation*}
%     \operatorname{im}(A P) \subseteq \operatorname{im}(P)+\operatorname{im}(B), \qquad \operatorname{im}(P)+\operatorname{ker}(C)=\mathbb{R}^n .
% \end{equation*}
% \begin{equation*}
% \operatorname{im}(D) \subseteq \operatorname{ker}(C), \qquad \operatorname{im}(P) \oplus \operatorname{im}(D)=\mathbb{R}^n ,
% \end{equation*}

% \begin{equation*}
% M P = I_m, \qquad I_n=P M+D E.
% \end{equation*}

By straightforward computations, this geometric design method proposed in Theorem~\ref{thm:abstractDesign} also satisfies~(\ref{equ:PiRelationV2}) with
\begin{equation}\label{equ:PiRelationNGamma}
    N = \left[\begin{array}{c}
    -\hat{L}M \\ E 
    \end{array}\right], \qquad 
    \Gamma = \left[\begin{array}{c}
    I_{m} \\ 0_{(\hat{m} - m) \times m}
    \end{array}\right].
\end{equation}
Note that this inherently implies that the dimension $\hat{m}$ of the input $v$ to the simpler system~(\ref{equ:ReducedSystem}) is larger than the dimension $m$ of the input $u$ to the large-scale system~(\ref{equ:ConcreteSystem}), \textit{i.e.}, $\hat{m} > m$. This is natural in the ASHC technique but not expected in MOR techniques.

Through Theorem~\ref{thm:PiRelationSyl}, the $M$-relation requirement can also be interpreted from the perspective of moment matching. In particular, this interpration is based on the interconnection structure formed by system~(\ref{equ:ReducedSystem}) and system~(\ref{equ:ConcreteSystem})
via the link~(\ref{equ:PiRelationLink}),
which is depicted in Fig.~\ref{fig:Interconnect_PiRelation}. Then we have the following result.
% brings the benefit that it provides an interconnection-based perspective by the linear relation 
% \begin{equation}\label{equ:PiRelationLink}
%    v = Nx + \Gamma u,
% \end{equation}
% converting~(\ref{equ:PiRelationState}) into Sylvester-equation-based requirements~(\ref{equ:PiRelationSyl}) and~(\ref{equ:PiRelationInput}). This translation suggests that, similarly to the bounded output discrepancy property, $M$-relation 

\begin{thm}\label{thm:PiRelationMatching}
    Consider the interconnection shown in Fig.~\ref{fig:Interconnect_PiRelation}. Suppose $u \equiv 0$ and $A$ has simple eigenvalues satisfying $\sigma(A) \subset \mathbb{C}_0$. Suppose $F$ is Hurwitz. Let $N \in \mathbb{R}^{\hat{m} \times n}$  be such that $(A,\, N)$ is observable. Let $M \in \mathbb{R}^{\hat{n} \times n}$ be such that~(\ref{equ:PiRelationSyl}) holds and let the pair $(A,\, x(0))$ be excitable. Then, there exists a one-to-one relation between the moment $HM$ of system~(\ref{equ:ReducedSystem}) at $(A,\, N)$ and the steady-state output $\psi$. 
\end{thm}
\begin{proof}
    The result follows directly from Theorem~\ref{thm:SteadyState}.
\end{proof}

% More precisely, a direct observation of~(\ref{equ:PiRelationV2}) is that equations~(\ref{equ:PiRelationSyl}) and~(\ref{equ:PiRelationOutputV2}) are ``swapped'' versions of~(\ref{equ:SimuCondition}), yielding the exact match of steady-state response of the output of system~(\ref{equ:ConcreteSystem}) and the output of a direct interconnection, shown in Fig.~\ref{fig:Interconnect_PiRelation} of the signal generator of the form~(\ref{equ:ReducedSystem}) with $B = 0$ and system~(\ref{equ:ReducedSystem})
% by the link $v = Nx + \Gamma u$ with $\Gamma = 0$. 

As the $M$-relation requirement induces the invariant set $\mathcal{M}_s = \{(x, \xi) \,|\, \xi = M x \}$, we can study, similarly to~(\ref{equ:directStateError}), the dynamics of the error $\varepsilon_s = \xi - M x$. The satisfaction of~(\ref{equ:PiRelationV2}) yields that
$\dot{\varepsilon}_s = F\xi + Gv - M Ax - M Bu = F\xi + GN x + G\Gamma u - (FM + GN)x - M Bu = F\varepsilon_s,$
% \begin{equation*}\label{equ:swappedStateError}
%     \begin{aligned}
%     \dot{\varepsilon}_s &= F\xi + Gv - M Ax - M Bu \\
%         &= F\xi + GN x + G\Gamma u - (FM + GN)x - M Bu\\
%         &= F\varepsilon_s,
%     \end{aligned}
% \end{equation*}
and the error in the output of two systems satisfies 
$
    \varepsilon_y = \psi - y = H\xi - Cx = H\varepsilon_s.
$
This result demonstrates that, starting with any initial condition $(x(0), \xi(0)) \in \mathcal{M}_s$, the $M$-relation yields the exact matching of the two output trajectories $y$ and $\psi$ in the direct interconnection shown in Fig.~\ref{fig:Interconnect_PiRelation} for any input $u$. Conversely, if $(x(0), \xi(0)) \not\in \mathcal{M}_s$, these output trajectories will match at the steady state only if $F$ is Hurwitz, \textit{i.e.}, only if the invariant set $\mathcal{M}_s$ is also attractive, as assumed by Theorem~\ref{thm:PiRelationMatching}.

\begin{figure}[tbp]
\begin{centering}
    \includegraphics[width=0.85\linewidth]{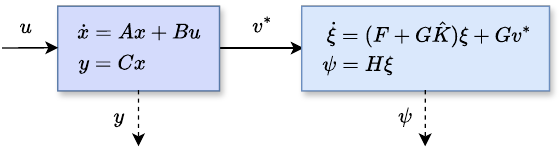}
    \caption{Interconnection for the interpretation of the linear matrix equations~(\ref{equ:PiRelationV2}) with a stabilising link $v^{*} = \Gamma u + (N - \hat{K}M)x$.}
    \label{fig:Interconnect_PiRelation_Stable}
\end{centering}
\end{figure}

If $F$ is not Hurwitz (as $\sigma(F)$ may be used to define the desired moments in Theorem~\ref{thm:hierarchicalMatching}), the input $v$ in~(\ref{equ:PiRelationLink}) can also be modified to render $\mathcal{M}_s$ both invariant and attractive for the same $M$. Assume the pair $(F,\, G)$ is stabilisable, \textit{i.e.}, there exists a matrix $\hat{K} \in \mathbb{R}^{\hat{m} \times \hat{n}}$ such that $\sigma(F + G\hat{K}) \subset \mathbb{C}_{< 0}$. Then the input 
\begin{equation}\label{equ:PiRelationLinkStable}
    v = Nx + \Gamma u + \hat{K}(\xi - M x)
\end{equation}
leads to $\dot{\varepsilon}_s = (F + G\hat{K}) \varepsilon_s$. Therefore, the invariant set $\mathcal{M}_s$ is exponentially attractive. This result shows that the $M$-relation alternatively means that for any state $x$ with any inputs $u$ into system~(\ref{equ:ConcreteSystem}), there exists an input $v$ to system~(\ref{equ:ReducedSystem}) such that the steady-state outputs of two systems, $\psi$ and $y$, are the same for all initial conditions $x(0)$ and $\xi(0)$. Consequently, a similar matching result as Theorem~\ref{thm:hierarchicalMatching} exists, which relaxes the stability assumption on system~(\ref{equ:ReducedSystem}) of Theorem~\ref{thm:PiRelationMatching}.

\begin{corol}\label{corol:PiRelationMatching}
    Consider the interconnection structure shown in Fig.~\ref{fig:Interconnect_PiRelation_Stable}. Suppose $u \equiv 0$, the pair $(F,\, G)$ is stabilisable, and $A$ has simple eigenvalues satisfying $\sigma(A) \subset \mathbb{C}_0$ and $\sigma(A) \cap \sigma(F) = \emptyset$. Let $N\in\mathbb{R}^{\hat{m}\times n}$ and $\hat{K}\in\mathbb{R}^{\hat{m} \times \hat{n}}$ in~(\ref{equ:PiRelationLinkStable}) be such that $(A,\, N)$ is observable and $\sigma(F + G\hat{K}) \subset \mathbb{C}_{< 0}$. Let $M \in \mathbb{R}^{\hat{n} \times n}$ be such that~(\ref{equ:PiRelationSyl}) holds and let the pair $(A,\, x(0))$ be excitable. Then, there exists a one-to-one relation between the moment $HM$ of system~(\ref{equ:ReducedSystem}) at $(A,\, N)$ and the steady-state output $\psi$ of system~(\ref{equ:ReducedSystem}) in the interconnection in Fig.~\ref{fig:Interconnect_PiRelation_Stable}.
\end{corol}
\begin{proof}
    The proof is similar to that of Theorem~\ref{thm:hierarchicalMatching} and, thus, is omitted.
\end{proof}

\begin{figure}[tbp]
\begin{centering}
    \includegraphics[width=0.8\linewidth]{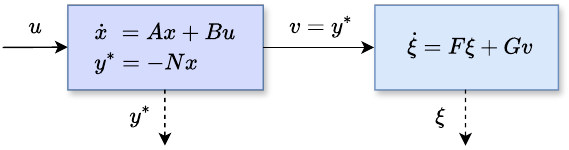}
    \caption{Swapped interconnection for the interpretation of the linear matrix equations~(\ref{equ:PiRelationV2}).}
    \label{fig:Interconnection_Swapped}
\end{centering}
\end{figure}

\begin{remark}
    By the previous derivations of the error dynamics $\varepsilon_s$ and $\varepsilon_y$, the assumption that the eigenvalues of $A$ are simple and satisfy $\sigma(A) \subset \mathbb{C}_0$ in Theorem~\ref{thm:PiRelationMatching} and Corollary~\ref{corol:PiRelationMatching} is only used for guaranteeing the existence of non-trivial steady-state responses of the outputs $y$ and $\psi$. This assumption is not necessary for $\varepsilon_s$ and $\varepsilon_y$ to converge to zero.
\end{remark}

% \begin{remark}\label{rmk:PiRelationSwapped}
%     It may also be possible to interpret the $M$-relation condition~(\ref{equ:PiRelationV2}) in Theorem~\ref{thm:PiRelationSyl} as a type of swapped interconnection introduced in Section~\ref{sec:Prelimiaries}. In fact, if we only look at the first two equations of~(\ref{equ:PiRelationV2}), assume system~(\ref{equ:ConcreteSystem}) is SISO, and let $\Gamma = 1$ (meaning that $\hat{m} = m = 1$) and $N = C$, then the satisfaction of~(\ref{equ:PiRelationSyl}) and~(\ref{equ:PiRelationInput}) suggests the one-to-one relation between the moments at $(F,\, G)$ and the steady-state response of the output of a swapped interconnection portrayed in Fig.~\ref{fig:Interconnection_Swapped}, which is formed by system~(\ref{equ:ConcreteSystem}) and a filter of form~(\ref{equ:ReducedSystem}) with $H = 0$, under the condition that $x(0) = \xi(0) = 0$ and the input $u$ is any non-zero exponentially-decaying signal. Additionally, when $H = 0$, the satisfaction of~(\ref{equ:PiRelationSyl}) and~(\ref{equ:PiRelationInput}) yields that system~(\ref{equ:ReducedSystem}) can also be seen as a special ``reduced-order model'' \red{with $F$, $G$, $H$ such that $H =0$ and $Q = F$, where~(\ref{equ:MMConditionSwapped}) is solved with $\Upsilon_r = I_m$.}
% \end{remark}

Now we provide an alternative interpretation. Rather than interpreting the $M$-relation requirement~(\ref{equ:PiRelationV2}) using the direct interconnections in Figs.~\ref{fig:Interconnect_PiRelation} and~\ref{fig:Interconnect_PiRelation_Stable}, it is possible to interpret this requirement as a type of swapped interconnection introduced in Section~\ref{sec:PrelimiariesMM} as shown in the next result.
\begin{thm}\label{thm:PiRelationMatchingSwapped}
    Consider the interconnection shown in Fig.~\ref{fig:Interconnection_Swapped}. Suppose $\sigma(A) \subset \mathbb{C}_{<0}$. Let $F$ have simple eigenvalues satisfying $\sigma(F) \subset \mathbb{C}_0$ and $\sigma(F) \cap \sigma(A) = \emptyset$. Assume the pair $(F,\, G)$ is reachable. Let $M \in \mathbb{R}^{\hat{n} \times n}$ be such that~(\ref{equ:PiRelationSyl}) holds. Then, there exists a one-to-one relation between the moment $MB$ of system
    \begin{equation}\label{equ:ConcreteSysVout}
        \dot{x} = A x+B u, \quad y^{*} = -N x,
    \end{equation}
    at $(F,\, G)$, with $y^{*}(t) \in \mathbb{R}^{\hat{m}}$, and the steady-state response of output $\xi$ with $x(0) = \xi(0) = 0$ for any non-zero signal $u$ that exponentially decays to zero. 
\end{thm}
\begin{proof}
    The result follows directly from Theorem~\ref{thm:SteadyState}.
\end{proof}

Note that the matrix $MB$ is not the moment of system~(\ref{equ:ConcreteSystem}), but that of system~(\ref{equ:ConcreteSysVout}), which flows as~(\ref{equ:ConcreteSystem}), but has a different output $y^{*} = -Nx$, where $-N$ is not necessarily equal to $C$. Moreover, Theorem~\ref{thm:abstractDesign} implies that a simpler system that satisfies the $M$-relation~(\ref{equ:PiRelationV2}) can be designed with $G=\left[MB \ \, MAD \right]$ and with $N$ and $\Gamma$ in~(\ref{equ:PiRelationNGamma}). Thus, system~(\ref{equ:ReducedSystem}) with $H = 0$ is the same as system~(\ref{eq-inv2}), for $Q=F$, $R=G$, $C=-N$ (and so $\Upsilon = M$) and $v=\Gamma u$. The resulting abstract system takes the form
$$
\dot \xi = F \xi + MB u.
$$
Thus, in a dual fashion with respect to the discussion after Theorem~\ref{thm:hierarchicalMatching}, the abstract system with input $v = \Gamma u$ behaves as a 
``\textit{limiting}'' reduced-order model. 
It should be stressed that the abstract system is not of the form (\ref{equ:ReducedMMQR}), because $H=0$, and consequently it does not match the moments $MB$ at $(F,G)$. But, as explained in Section~\ref{sec:PrelimiariesMM}, the limiting system~(\ref{eq-inv2}) is the system towards which the reduced-order model~(\ref{equ:ReducedMMQR}) tends to when $Q-RH$ is Hurwitz. Thus, even though the abstract system is not a reduced-order model by moment matching, it still has the property of matching, for any non-zero signal $u$ that exponentially decays to zero, the steady-state output of the swapped interconnection in Fig.~\ref{fig:Interconnection_Swapped}.

%In this case, inspired by~(\ref{equ:ReducedMMQR}) in the moment matching perspective, system~(\ref{equ:ReducedSystem}) with $H = 0$, which serves as a filter in Fig.~\ref{fig:Interconnection_Swapped}, can also be seen as a special ``reduced-order model'' of system~(\ref{equ:ConcreteSysVout}) when subject to the ``same'' input $v = \Gamma u = [u^{\top} \, \ 0_{1 \times (\hat{m} - m)}]^{\top}$ as system~(\ref{equ:ConcreteSystem}).

% In fact, the first two equations of~(\ref{equ:PiRelationV2}), assume system~(\ref{equ:ConcreteSystem}) is SISO, and let $\Gamma = 1$ (meaning that $\hat{m} = m = 1$) and $N = C$, then the satisfaction of~(\ref{equ:PiRelationSyl}) and~(\ref{equ:PiRelationInput}) suggests the one-to-one relation between the moments at $(F,\, G)$ and the steady-state response of the output of a swapped interconnection portrayed in Fig.~\ref{fig:Interconnection_Swapped}, which is formed by system~(\ref{equ:ConcreteSystem}) and a filter of form~(\ref{equ:ReducedSystem}) with $H = 0$, under the condition that $x(0) = \xi(0) = 0$ and the input $u$ is any non-zero exponentially-decaying signal. Additionally, when $H = 0$, the satisfaction of~(\ref{equ:PiRelationSyl}) and~(\ref{equ:PiRelationInput}) yields that system~(\ref{equ:ReducedSystem}) can also be seen as a special ``reduced-order model'' \red{with $F$, $G$, $H$ such that $H =0$ and $Q = F$, where~(\ref{equ:MMConditionSwapped}) is solved with $\Upsilon_r = I_m$.}

\subsection{Brief Summary of the Two Key Requirements}
We briefly summarise the interpretations presented above. The ASHC technique consists of designing a simpler system and an interface such that: i) the output responses of the two systems have bounded discrepancy; ii) the simpler system is $M$-related to the large-scale system. Regarding requirement i), we have the following interpretations.
\begin{itemize}
    \item When the input $v \equiv 0$, we provide conditions under which the steady-state output of system~(\ref{equ:ConcreteSystem}) in the direct interconnection shown in Fig.~\ref{fig:Interconnect_Output} is matched with the steady-state output of system~(\ref{equ:ReducedSystem}). In other words, the steady-state output response of system~(\ref{equ:ConcreteSystem}) admits a one-to-one relation with the moments of system~(\ref{equ:ConcreteSystem}) at $(F,\, \hat{L})$, see Theorem~\ref{thm:hierarchicalMatching}.
    \item When the input $v \not\equiv 0$ but $v \in L^\infty$, the mismatch between the steady-state outputs of the two systems, interconnected via the interface, is bounded and can be minimized.
\end{itemize}
Regarding requirement ii), we have the following interpretations.
\begin{itemize}
    \item For any bounded input $u$, we provide conditions under which the steady-state output of system~(\ref{equ:ReducedSystem}) in the direct interconnections shown in Fig.~\ref{fig:Interconnect_PiRelation} and Fig.~\ref{fig:Interconnect_PiRelation_Stable} is matched with the steady-state response of system~(\ref{equ:ConcreteSystem}). In other words,  the steady-state output response of system~(\ref{equ:ReducedSystem}) admits a one-to-one relation with the moments of system~(\ref{equ:ReducedSystem}) at $(A,\, N)$, see Theorem~\ref{thm:PiRelationMatching} and Corollary~\ref{corol:PiRelationMatching}.
    % \item When the input $u \not\equiv 0$, there exists an input $v$ to system~(\ref{equ:ReducedSystem}) such that the mismatch between the steady-state outputs of system~(\ref{equ:ConcreteSystem}) and system~(\ref{equ:ReducedSystem}) \red{is zero for all time.}
    \item When the input $u$ is any non-zero signal that exponentially decays to zero, we provide conditions under which the steady-state output of the swapped interconnection shown in Fig.~\ref{fig:Interconnection_Swapped} admits a one-to-one relation with the moments of system~(\ref{equ:ConcreteSysVout}) at $(F,\,G)$, see Theorem~\ref{thm:PiRelationMatchingSwapped}.
\end{itemize}

In summary, the final obtained abstract system, for $v=\Gamma u$, namely
\begin{equation}
\label{equ:final-abs}
    \dot{\xi} = F \xi + MB u, \qquad
    \psi = C P \xi,
\end{equation}
is reminiscent of the two-sided reduced-order model by moment matching, in the ``\textit{limiting}'' sense described in the previous sections.
%Building on the above interpretations, a key distinction between the two \red{requirements} in the ASHC \red{technique} and the requirement of moment matching is that the simpler system serves as both a ``reduced-order model'' when $G = 0$ and a signal generator, and also both a ``reduced-order model'' when $H = 0$ and a filter. In fact, a simpler system in the ASHC \red{technique} does not match the moments of the original large-scale system \red{if $F$, $G$, and $H$ that simultaneously solve~(\ref{equ:SimuCondition}) and~(\ref{equ:PiRelationV2}) are all non-zero.}

% In addition, \GS{we} have proved that the ASHC \red{technique} is solved if we can construct a simpler system~(\ref{equ:ReducedSystem}) and an interface~(\ref{equ:interface}) such that there exist \red{auxiliary matrices that solve}~(\ref{equ:SimuCondition}) and~(\ref{equ:PiRelationV2}). However, it may be difficult to find a solution for these equations with $\hat{n} < n$. The geometric method in Theorem~\ref{thm:abstractDesign} proposed by~\cite{ref:girard2009hierarchical} found the simpler model by increasing the dimension of inputs\red{, \textit{i.e.}, $\hat{m} > m$.} In some other work, the \red{requirement}~(\ref{equ:SimuCondition}) was modified to enhance the solvability of the \red{problem, see,} \textit{e.g.},~\cite{ref:smith2020approximate} for more details.

\subsection{Discussion}
The importance of this paper lies in bridging the conceptual gap between moment matching and the ASHC technique. The moment matching method has an advantage of admitting natural extensions beyond LTI systems, such as nonlinear systems~\cite{ref:shakib2024optimal}, time-delay systems~\cite{ref:scarciotti2015model}, and stochastic systems~\cite{ref:scarciotti2021moment}. Meanwhile, other techniques, such as data-driven methods, have also been proposed for solving moment matching problems~\cite{ref:scarciotti2017data,ref:mao2024data}.
Those directions have been partly, but not extensively, explored in the ASHC context. Bridging the gap between these two techniques suggests that the existing results in the moment matching field can potentially lead to extensions of the ASHC method to more generic settings. Conversely, the current developments in the ASHC technique, \textit{e.g.}, compositional method for constructing abstractions for networked systems~\cite{ref:rungger2016compositional,ref:smith2020approximate} and symbolic control method for implementing approximate simulations~\cite{ref:tabuada2008approximate}, do not have a counterpart in the moment matching framework.

To give a concrete hint of this potential, we consider the development of a nonlinear ASHC framework. Since we have characterised the objects that define the final abstraction \eqref{equ:final-abs} by means of the Sylvester equations (\ref{equ:SimuCondition}) and (\ref{equ:PiRelationV2}), by leveraging the nonlinear moment matching framework, a nonlinear abstraction can be characterised by means of the invariance equations \cite[(28) and (35)]{ref:scarciotti2024interconnection} \textit{mutatis mutandis}. A work in this direction is already in preparation.
% Such developments could open the door for, \textit{e.g.}, model predictive control techniques for solving the ASHC problem of networked systems~\cite{ref:smith2020approximate}, \red{in a} moment matching perspective.

%Another observation is that the ``best'' simulation function should leave a small error in the control synthesis of the large-scale system. In the scope of moment matching, similarly, ~\cite{ref:shakib2024optimal} and~\cite{ref:necoara2020h_2} have also explored minimizing the error between responses of the original and the reduced-order model.

% In addition, moment matching was originally developed for large-scale systems of dimension greater than $10000$, while the ASHC problem typically requires solving linear matrix inequalities like~(\ref{equ:simuLMI}) in the design process and is therefore limited to complex systems with a few hundred states. Bridging the gap can also help extend the ability of the ASHC technique to deal with more complex systems.

\begin{figure}[tbp]
\begin{centering}
    \includegraphics[width=0.9\linewidth]{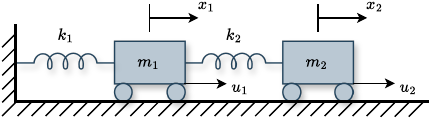}
    \caption{Two-spring two-mass mechanical system.}
    \label{fig:example}
\end{centering}
\end{figure}

\section{Numerical Example}\label{sec:example}
We illustrate the theory presented in the paper by means of a practical example. We study the ASHC problem of system~(\ref{equ:ConcreteSystem}) with matrices
\begin{equation*}
\begin{aligned}
    A &=
\left[
\begin{array}{cccc}
    0  &  0  &  1  &  0 \\
    0  &  0  &  0  &  1 \\
   \!\!\!\!-\frac{k_1+k_2}{m_1}  & \frac{k_2}{m_1}  &  0  &  0 \\
    \frac{k_2}{m_2}  & \!\!\!\!-\frac{k_2}{m_2}  & 0  &  0
\end{array}
\right], \quad
B =
\left[
\begin{array}{cc}
    0  &  0 \\
    0  &  0 \\
    \frac{1}{m_1}  &  0 \\
    0  &  \frac{1}{m_2} \\ 
\end{array}
\right],\!\!
% C &=
% \left[
% \begin{array}{cccc}
%   1  & 0  & 0  & 0 \\
%   0  & 1  & 0  & 0
% \end{array}
% \right],
\end{aligned} 
\end{equation*}
and $C = [I_2, 0_{2 \times 2}]$, where the system dimensions $n = 4$, $m = p = 2$. This concrete system models the two-spring two-mass mechanical system in Fig.~\ref{fig:example}, which consists of two masses $m_1$ and $m_2$ connected by two springs with spring coefficients $k_1$ and $k_2$. Variables $x_1$ and $x_2$ denote the displacements from the equilibrium positions of the masses $m_1$ and $m_2$ under the control forces $u_1$ and $u_2$, respectively. Then the input, output, and state variables of system~(\ref{equ:ConcreteSystem}) are denoted as $u = [u_1, u_2]^\top$, $y = [x_1, x_2]^\top$, and $x = [x_1, x_2, x_3, x_4]^\top$ with $x_3 = \dot{x}_1$ and $x_4 = \dot{x}_2$. In the following simulation, we select the parameters as $k_1 = 100$ N/m, $k_2 = 50$ N/m, $m_1 = 20$ kg, and $m_2 = 10$ kg. Then by following the design method proposed by~Theorem~\ref{thm:abstractDesign}, a valid simpler model~(\ref{equ:ReducedSystem}) solving the ASHC problem is given by the selection
\begin{equation*}
\begin{aligned}
    F &=
\left[\!\!
\begin{array}{cc}
     0  & 10 \\
    -10  & 0 
\end{array}
\!\!\right]\!, \quad
G =
\left[\!\!
\begin{array}{rrrr}
  0  & 0  &  1  &  0 \\
  0  & 0  &  0  &  1
\end{array}
\!\!\right]\!, \quad
H =
\left[\!\!
\begin{array}{cc}
  1  &  0 \\
  0  &  1
\end{array}
\!\!\right]\!,
\end{aligned} 
\end{equation*}
with dimensions $\hat{n} = 2$ and $\hat{m} = 4$. In addition, the matrices 
% $P$, $\hat{L}$, $M$, and $D$ given by
\begin{equation*}
\begin{aligned}
    P &=
\left[
\begin{array}{cccc}
  1  &  0  &  0  &  -10 \\
  0  &  1  &  10  &  0
\end{array}
\right]^{\top}, \qquad
\hat{L}=
\left[
\begin{array}{cc}
  -1850  &  -50 \\
  -50  &  -950
\end{array}
\right],
% M &=
% \left[
% \begin{array}{cccc}
%   1  &  0  &  0  &  0 \\
%   0  &  1  &  0  &  0
% \end{array}
% \right], \qquad \ \,
% D = \left[
% \begin{array}{cccc}
%   0  &  0  &  1  &  0 \\
%   0  &  0  &  0  &  1
% \end{array}
% \right]^{\top},
\end{aligned} 
\end{equation*}
$M = [I_2, 0_{2 \times 2}]$, and $D = [0_{2 \times 2}, I_2]^\top$ solve~(\ref{equ:SimuCondition}) and satisfy the design requirements in Theorem~\ref{thm:abstractDesign}. The matrix $\hat{R}$ is selected with all entries equal to one. Furthermore, the matrices $N$ and $\Gamma$ are given by
\begin{equation*}
    N =
    \left[
    \begin{array}{cccc}
     -1850  &  50 &  0  &  0 \\
      50  &  950   &  0  &  0 \\
       0   & -10  &  1  &  0 \\
      10   &  0   &  0  &  1
    \end{array}
    \right], \quad \Gamma =
    \left[
    \begin{array}{cc}
      1  &  0 \\
      0  &  1 \\
      0  &  0 \\
      0  &  0
    \end{array}
    \right],
\end{equation*}
solving~(\ref{equ:PiRelationV2}) with $M$ as above. Note that the matrices $A$ and $F$ both have simple eigenvalues located on the imaginary axis, namely $\sigma(A) = \{\pm 3.1623i, \pm 1.5811i \}$ and $\sigma(F) = \{\pm 10i\}$. In addition, it can be verified that the pairs $(A,\,B)$ and $(F,\, G)$ are stabilisable. In what follows, we illustrate our interpretation of the ASHC technique in Section~\ref{sec:compare} by implementing the two interconnection frameworks shown in Figs.~\ref{fig:Interconnect_Output} and~\ref{fig:Interconnect_PiRelation_Stable}. The initial conditions of the two systems are randomly selected as $
x_0 = [6.7794,\, -1.3348,\, -0.5875,\, 1.2143]^{\top}$ and $\xi_0 = [-4.2811,\, 0.8733]^{\top}$.

\begin{figure}[tbp]
\begin{centering}
    \includegraphics[width=0.93\linewidth]{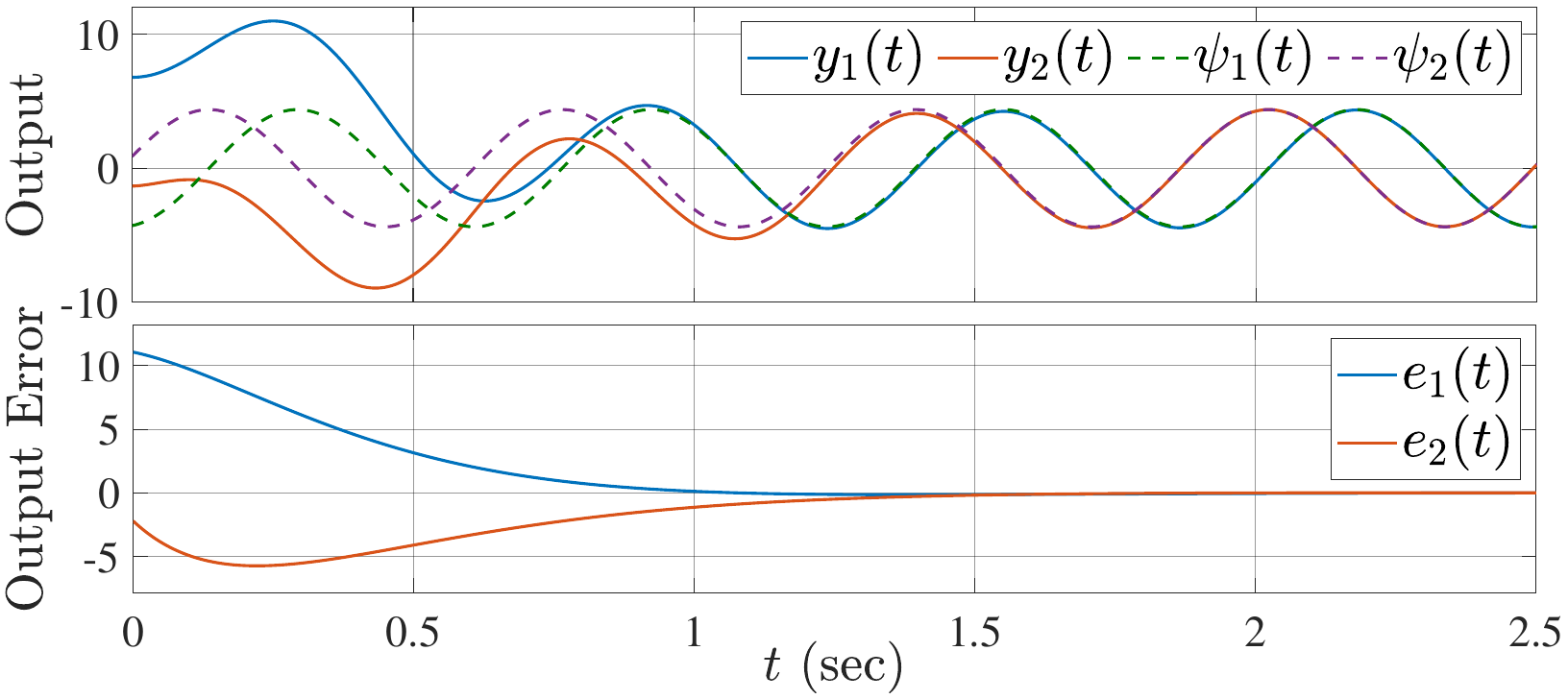}
    \caption{Time histories of outputs $y = [y_1, y_2]^\top$ (solid) and $\psi = [\psi_1, \psi_2]^\top$ (dashed) of systems~(\ref{equ:ConcreteSystem}) and~(\ref{equ:ReducedSystem}) in the interconnection in Fig.~\ref{fig:Interconnect_Output} (top) and their error $e_y = [e_1, e_2]^\top = y - \psi$ (bottom) with $v \equiv 0$.}
    \label{fig:Simulation_OutputZero}
\end{centering}
\end{figure}
\begin{figure}[tbp]
\begin{centering}
    \includegraphics[width=0.93\linewidth]{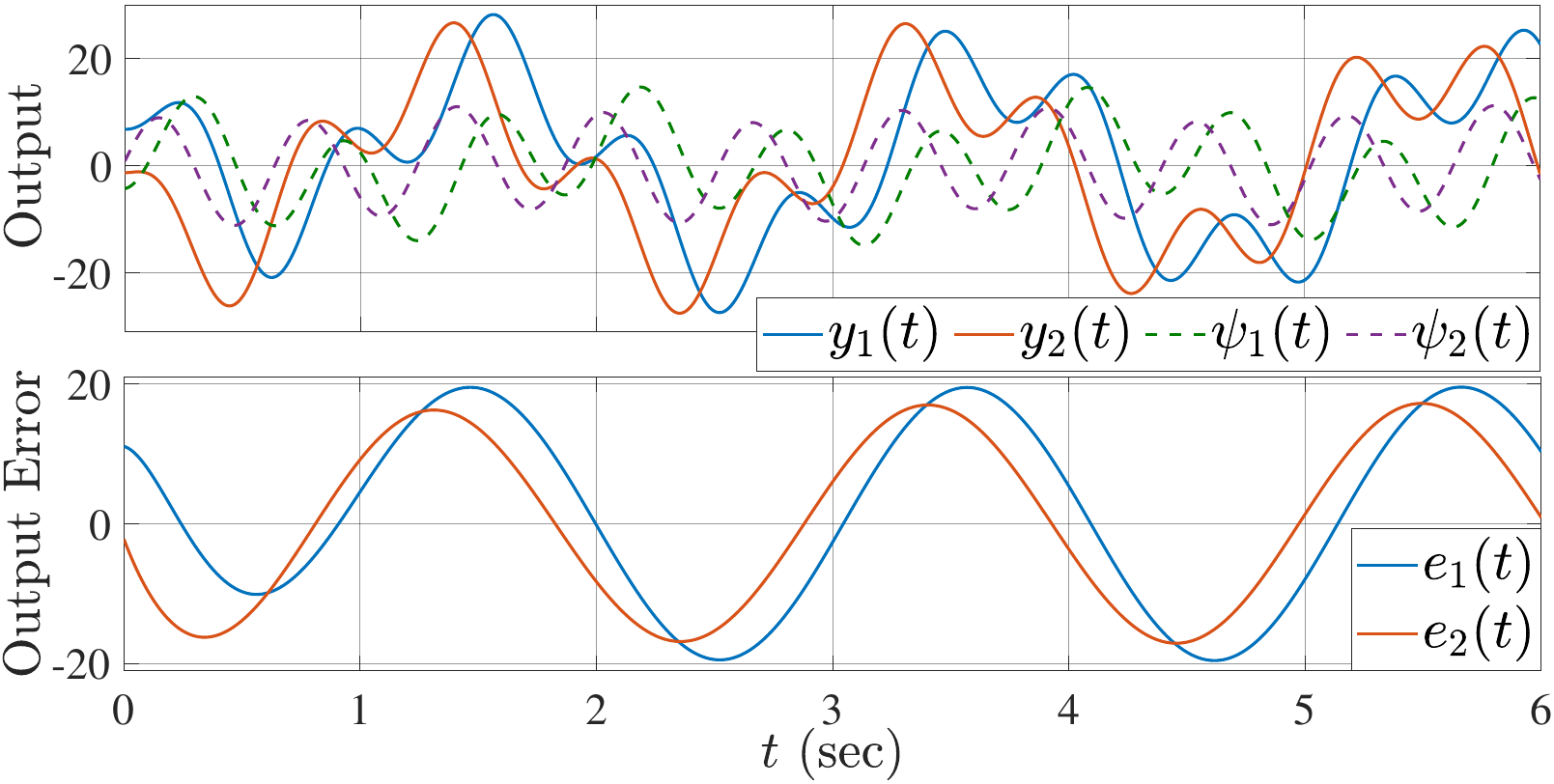}%height=4.3cm
    \caption{Time histories of outputs $y = [y_1, y_2]^\top$ (solid) and $\psi = [\psi_1, \psi_2]^\top$ (dashed) of systems~(\ref{equ:ConcreteSystem}) and~(\ref{equ:ReducedSystem}) in the interconnection in Fig.~\ref{fig:Interconnect_Output} (top) and their error $e_y = [e_1, e_2]^\top = y - \psi$ (bottom) with $v \not\equiv 0$.}
    \label{fig:Simulation_OutputNonZero}
\end{centering}
\end{figure}

We first implement the interconnection shown in Fig.~\ref{fig:Interconnect_Output}, which is equivalent to the hierarchical structure used in the ASHC technique, and is used for interpreting the bounded output discrepancy requirement. The stabilising gain $K$ is selected such that $\sigma(A + BK) = \{-3 \pm 1.5i,\, -5 \pm 2i,\}$. %Meanwhile, as the paper does not aim to address the performance of the simulation function, \red{we arbitrarily selected $\hat{R} \in \mathbb{R}^{m \times \hat{m}}$ with all entries equal to one.}
% the matrix $\hat{R}$ in~(\ref{equ:interface}) is simply randomly generated as
% \begin{equation*}
%     \hat{R} =\left[
%     \begin{array}{rrrr}
%       2.2016  &  0.1699  & -3.4350  &  1.9480 \\
%      -1.5310  &  0.5669  &  0.6206  & -0.7354
%     \end{array}
%     \right].
% \end{equation*}
Then we conduct simulations with $v \equiv 0$ and $v \not\equiv 0$ separately, where, when $v \not\equiv 0$, we arbitrarily choose $v = [51.032\sin(4t),\, -25.945\operatorname{sign}(\sin (6t)),\, 0,\, 48.056\cos(3t)]^{\top}$. The results are shown in Figs.~\ref{fig:Simulation_OutputZero} and~\ref{fig:Simulation_OutputNonZero}. Fig.~\ref{fig:Simulation_OutputZero} displays for $v \equiv 0$ the time histories (top) of outputs $y$ and $\psi$ of the interconnection in Fig.~\ref{fig:Interconnect_Output} and (bottom) of their error $e_y = y - \psi$. The plots indicate that when $v \equiv 0$, the steady-state outputs of (the stabilised) system~(\ref{equ:ConcreteSystem}) in the interconnection in Fig.~\ref{fig:Interconnect_Output} and system~(\ref{equ:ReducedSystem}) coincide with each other, suggesting that the moment matching result presented by Theorem~\ref{thm:hierarchicalMatching} holds. The outputs and the corresponding error when $v \not\equiv 0$ are shown in Fig.~\ref{fig:Simulation_OutputNonZero}. In this case, the error between these two output trajectories is bounded, as anticipated by the results in Section~\ref{sec:compareOutputBound}.

\begin{figure}[tbp]
\begin{centering}
    \includegraphics[width=0.93\linewidth]{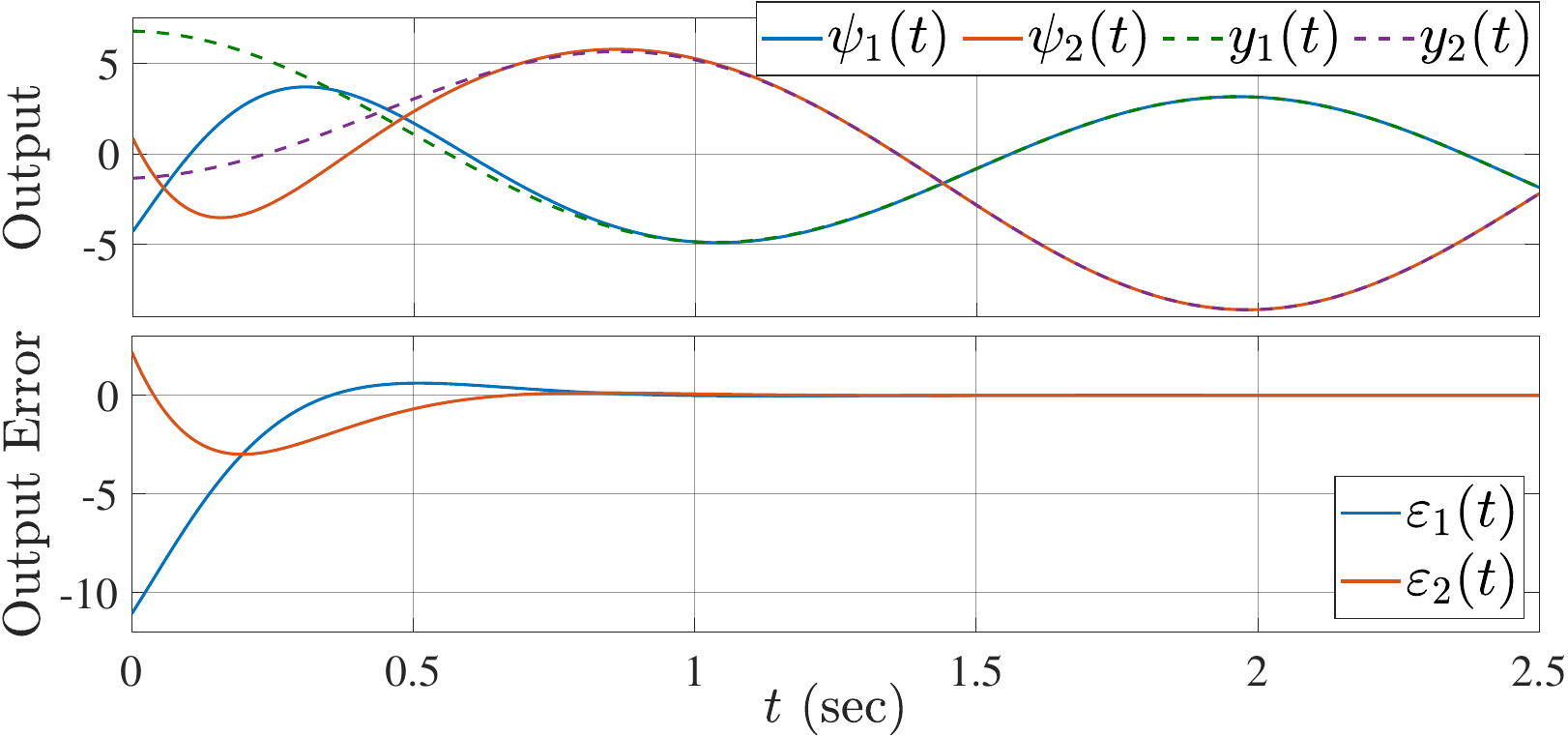}
    \caption{Time histories of outputs $y = [y_1, y_2]^\top$ (dashed) and $\psi = [\psi_1, \psi_2]^\top$ (solid) of systems~(\ref{equ:ConcreteSystem}) and~(\ref{equ:ReducedSystem}) in the interconnection in Fig.~\ref{fig:Interconnect_PiRelation_Stable} (top) and their error $\varepsilon_y = [\varepsilon_1, \varepsilon_2]^\top = \psi - y$ (bottom) with $u \equiv 0$.}
    \label{fig:Simulation_PiZero}
\end{centering}
\end{figure}
\begin{figure}[tbp]
\begin{centering}
    \includegraphics[width=0.93\linewidth]{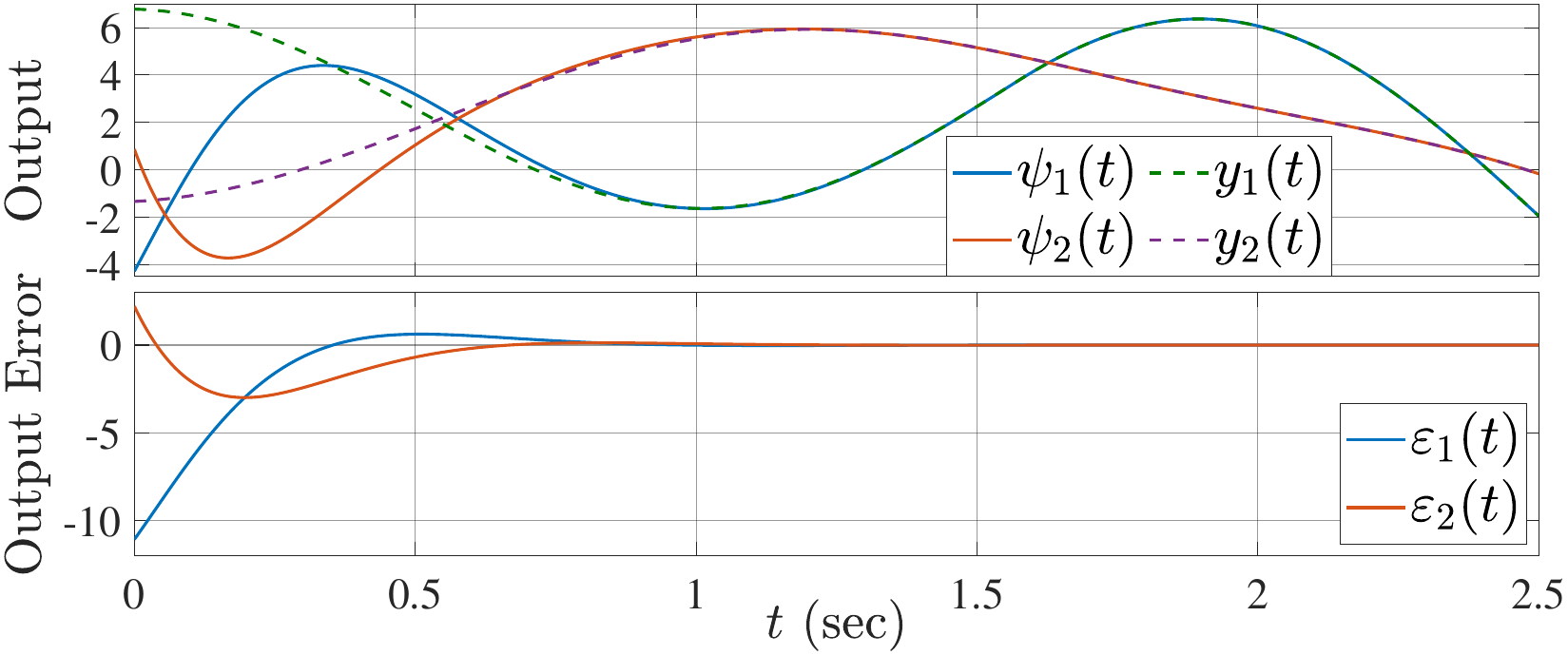}
    \caption{Time histories of outputs $y = [y_1, y_2]^\top$ (dashed) and $\psi = [\psi_1, \psi_2]^\top$ (solid) of systems~(\ref{equ:ConcreteSystem}) and~(\ref{equ:ReducedSystem}) in the interconnection in Fig.~\ref{fig:Interconnect_PiRelation_Stable} (top) and their error $\varepsilon_y =[\varepsilon_1, \varepsilon_2]^\top = \psi - y$ (bottom) with $u \not\equiv 0$.}
    \label{fig:Simulation_PiNonZero}
\end{centering}
\end{figure}

We finally implement the interconnection in Fig.~\ref{fig:Interconnect_PiRelation_Stable} that is used for interpreting the $M$-relation requirement, which states that for any $u$, there exists an input $v$ such that the output trajectories match. Again, we simulate the interconnection with $u \equiv 0$ and $u \not\equiv 0$ (arbitrarily set as $u =  [296.881\operatorname{sign}(\sin(2t)),\, -161.659\cos(3t)]^{\top}$) separately, with the stabilising gain $\hat{K}$ selected such that $\sigma(F + G\hat{K}) = \{-5 \pm 5i\}$. %\red{, where,} when $u \not\equiv 0$, we arbitrarily set \red{$u =  [296.881\operatorname{sign}(\sin(2t)),\, -161.659\cos(3t)]^{\top}$}. 
The results of the simulations are depicted in Figs.~\ref{fig:Simulation_PiZero} and~\ref{fig:Simulation_PiNonZero}, which show the time histories of (top) outputs $\psi$ and $y$ from the interconnection in Fig~\ref{fig:Interconnect_PiRelation_Stable} as well as (bottom) their difference $\varepsilon_y = \psi - y$ when $u \equiv 0$ (Fig.~\ref{fig:Simulation_PiZero}) and $u \not\equiv 0$ (Fig.~\ref{fig:Simulation_PiNonZero}). The figures imply that when system~(\ref{equ:ReducedSystem}) is $M$-related to system~(\ref{equ:ConcreteSystem}) with requirement~(\ref{equ:PiRelationV2}) satisfied by a stabilisable pair $(F, G)$, then $v$ as in~(\ref{equ:PiRelationLinkStable}) results in the matching of the steady-state outputs of the two (stabilised) systems in the interconnection in Fig.~\ref{fig:Interconnect_PiRelation_Stable}, no matter whether $u \equiv 0$ or not. This result coincides with the discussions in Section~\ref{sec:comparePiRelation}.

\section{Conclusion}\label{sec:concl}
In this paper, we have established a connection between moment matching and the ASHC techniques. %To build the conceptual bridge, we have looked at the two requirements of the ASHC technique through the lens of moment matching, showing that both the bounded output discrepancy requirement and the $M$-relation requirement can be characterised by Sylvester equations and, therefore, can be interpreted from the moment matching perspective based on certain interconnection structures. The results imply that, in the LTI case, these two techniques are deeply related under additional conditions. 
We have also identified research directions that enable the use of moment-matching-based techniques in ASHC problems and the use of ASHC techniques in moment matching problems. These results are currently under development. %Consequently, our future work will focus on leveraging the nonlinear/data-driven moment matching ideas to solve generalised ASHC problems.

\addtolength{\textheight}{-12cm}   % This command serves to balance the column lengths

%%%%%%%%%%%%%%%%%%%%%%%%%%%%%%%%%%%%%%%%%%%%%%%%%%%%%%%%%%%%%%%%%%%%%%%%%%%%%%%%

%%%%%%%%%%%%%%%%%%%%%%%%%%%%%%%%%%%%%%%%%%%%%%%%%%%%%%%%%%%%%%%%%%%%%%%%%%%%%%%%



% \section*{ACKNOWLEDGMENT}

% The preferred spelling of the word ÒacknowledgmentÓ in America is without an ÒeÓ after the ÒgÓ. Avoid the stilted expression, ÒOne of us (R. B. G.) thanks . . .Ó  Instead, try ÒR. B. G. thanksÓ. Put sponsor acknowledgments in the unnumbered footnote on the first page.

% %%%%%%%%%%%%%%%%%%%%%%%%%%%%%%%%%%%%%%%%%%%%%%%%%%%%%%%%%%%%%%%%%%%%%%%%%%%%%%%%

% References are important to the reader; therefore, each citation must be complete and correct. If at all possible, references should be commonly available publications.

\bibliographystyle{IEEEtran}
\bibliography{mybib}

\end{document}